\documentclass[reqno]{amsart}
\usepackage{graphicx,graphics}
\usepackage{subfig}
\usepackage{bm}

\newcommand{\D}{\mathcal{D}}

\newcommand{\Real}{\mathbb R}
\newcommand{\abs}[1]{\left\vert#1\right\vert}

\theoremstyle{definition}

\newtheorem{lem}{Lemma}
\newtheorem{thm}{Theorem}

\newtheorem{coll}{Corollary}
\newtheorem{rem}{Remark}

\DeclareMathOperator\erf{erf}

\numberwithin{thm}{section}
\numberwithin{lem}{section}
\numberwithin{coll}{section}
\numberwithin{rem}{section}
\numberwithin{exm}{section}
\numberwithin{prop}{section}
\numberwithin{equation}{section}

\numberwithin{equation}{section}

\begin{document}

\centerline {\textsc {\large A note on power generalized extreme value distribution}} \centerline {\textsc {\large and its properties}}
\vspace{0.5in}

\begin{center}
 Ali Saeb\footnote{Corresponding author: ali.saeb@gmail.com} \\
 Faculty of Mathematics, K. N. Toosi University of Technology,\\ Tehran-161816315, Iran

\vspace{.1in}
\end{center}
\vspace{1in}

\noindent {\bf Abstract:} 
Similar to the generalized extreme value (GEV) family, the generalized extreme value distributions under power normalization are introduced by Roudsari (1999) and Barakat et al. (2013). In this article, we study the asymptotic behavior of GEV laws under power normalization and derive expressions for the $k$th moments, entropy, ordering in dispersion, rare event estimation and application of real data set. We also show that, under some conditions, the Shannon entropy and variance of GEV families are ordered.

\vspace{0.5in}

\vspace{0.2in} \noindent {\bf Keywords:} Generalized extreme value family, power normalization, variance and entropy in ordering distributions, Bayesian modeling.

\vspace{0.5in}

\vspace{0.2in} \noindent {\bf MSC 2010 classification:} 62H10, 60G70

\newpage
\section{Introduction}
Let $X_1, X_2,\ldots, X_n$ is a sequence of independent and identically distributed (iid) random variables (rvs) with distribution function (df) $F.$ If, for some non-degenerate df $G,$ a df $F$ belongs to the max domain of attraction of $G$ under linear normalization and it denotes by  $F\in\D_\ell(G),$ then for some norming constants $a_n>0$ and $b_n\in\Real$
\begin{equation}
\lim_{n\to\infty}\Pr \left(\bigvee_{i=1}^{n}X_i\leq a_nx+b_n\right)=\lim_{n\to\infty}F^n \left(a_nx+b_n\right)=G(x).\label{Introduction_e1}
\end{equation}
Limit df $G$ satisfying (\ref{Introduction_e1}) are the well known generalized extreme value (GEV) distribution, namely,
\begin{equation}\label{gev}	G_{\tilde{\xi}}(x;\tilde{\mu},\tilde{\sigma})=\exp\left(-\left(1+\frac{\tilde{\xi}}{\tilde{\sigma}} (x-\tilde{\mu})\right)^{-1/\tilde{\xi}}_+\right),\end{equation}
 where, $\tilde{\xi}\in\Real/\{0\},$ $\tilde{\mu}\in\Real$ and $\tilde{\sigma}>0.$  The subset of the GEV family with $\tilde{\xi}=0$ is interpreted as the limit of (\ref{gev}) as $\tilde{\xi}\to 0,$ leading to the Gumbel family with df \[G(x;\tilde{\mu},\tilde{\sigma})=\exp\left(-\exp\left(-\frac{x-\tilde{\mu}}{\tilde{\sigma}}\right)\right),\;\;x\in\Real.\]
 
Criteria for $F\in \mathcal{D}_\ell(G)$ are described, for example, in the books of Embrechts et al. (1997) and de Haan and Ferreira (2006). Coles (2001) is good reference to the application of GEV distribution.

Pancheva (1984) studies limit laws of partial maxima of iid rvs under power normalization. Namely, $K$ is called p-max stable law and $F$ belongs to the p-max domain of attraction of $K$ under power normalization and denote it by $F\in \D_p(K),$ if for some $\delta_n>0,$ $\beta_n>0$
\begin{eqnarray}\label{pmax_lim}
\lim_{n\to\infty}\Pr\left(\left(\dfrac{\abs {\bigvee_{i=1}^{n}X_i}}{\delta_n}\right)^{1/\beta_n}\text{sign}\left(\bigvee_{i=1}^{n}X_i\right)\leq x\right)=K(x).
\end{eqnarray}
The limit laws $K$ satisfying (\ref{pmax_lim}) are the six p-max stable laws which we represent them in appendix \ref{pmax}.
Mohan and Ravi (1993) show that if a df $F\in \mathcal{D}_\ell (G)$ then there always exists a p-max stable law $K$ such that $F \in \D_p (K)$ and the converse need not be true always. They also investigate, the p-max stable laws attract more dfs to their max domains than the $\ell-$max stable laws. See also Christoph and Falk (1996) and Falk et al. (2004) for properties of dfs to belong to the p-max domain of attraction.

Roudsari (1999) demonstrates that the six p-max stable laws can be represented as two families. We call them log-GEV distribution in positive support and negative log-GEV distribution in negative support. 
Suppose a positive rv $X^+$ is said to have the log-GEV with location, scale and shape parameters $\mu\in\Real,$  $\sigma>0$ and $\xi\in\Real /\{0\}$ if its df is given by
\begin{eqnarray}\label{df1}	L_{1,\xi}(x;\mu,\sigma)=\left\lbrace	
\begin{array}{l l}	\exp\left(-\left(1+\frac{\xi}{\sigma}\log(e^{-\mu} x)\right)_+^{-1/\xi}\right), &\text{for, } \xi\neq 0;\\ \exp\left(-\left(xe^{-\mu}\right)^{-\frac{1}{\sigma}}_+\right), & \text{for, }\xi=0,\\
\end{array}\right.
\end{eqnarray}
where, $y_+=\max(0,y),$ and define a negative rv $X^-$ with df of negative log-GEV, if its df is given by,
\begin{eqnarray}\label{df2}	L_{2,\xi}(x;\mu,\sigma)=\left\lbrace		\begin{array}{ll}	\exp\left(-\left(1-\frac{\xi}{\sigma}\log(e^{-\mu}\abs{x})\right)_+^{-1/\xi}\right), &\text{for } \xi\neq 0;\\	\exp\left(-(\abs{x}e^{-\mu})^{\frac{1}{\sigma}}_+\right),& \text{for }\xi= 0.\\		
\end{array}\right.
\end{eqnarray}
The summarization these two families as a single one is easier to implement.
In other words, the unification of the log-GEV and negative log-GEV families into single family and it is called the power generalized extreme value (PGEV) family.
Suppose a rv $X$ is said to have the PGEV distribution with three parameters $\mu\in\Real,$ $\sigma>0$ and $\xi\in\Real/\{0\}$ if its df is given by
\begin{eqnarray}\label{df3}
L_{\xi}(x;\mu,\sigma)=\exp\left(-\left(1+\frac{\xi}{\sigma}\log(e^{-\mu}\abs{x})\text{sign}(x)\right)_+^{-1/\xi}\right),
\end{eqnarray}
The limit of (\ref{df3}) as $\xi\to 0,$ tending to the GEV distribution with  $\tilde{\sigma}=e^{\mu}$ and $\tilde{\xi}=\sigma\,\text{sign}(x).$ The df of PGEV for $\xi=0$ is well known in (\ref{gev}).
Barakat et al. (2013) study the statistical inference about the PGEV.
In appendix \ref{fig}, illustrate the density functions and confidence interval for quantile estimator of PGEV family and gives the figures \ref{fig2} and \ref{fig3} of standardized density function of $l_{\xi}$ for different values of $\xi.$

In this article, we obtain some mathematical properties of PGEV family and discuss maximum likelihood estimation of parameters and estimate the rare event by using the Bayesian method. We also, show that the PGEV has big variance and entropy in the class of extreme value distributions. The article is outlined as follows. In section 2, we first study the asymptotic behavior of generalized extreme value distributions under power normalization and we also, derive expressions of $k$th moments, the Shannon entropy and ordering in dispersion of GEV families. Maximum likelihood estimation, Bayesian modeling and illustrates the importance of the PGEV through the analysis of real data set are addressed in section 3. We provide the some calculating, plots and tables in appendices. 

Throughout the manuscript  $\gamma=-\int_{0}^{\infty}\log x \,e^{-x}dx$ denotes the Euler  constant with value $0.577\ldots$ and $\Gamma^{(k)}(\cdot)$ is $k$th derivative of gamma function. The inverse function of $h(\cdot)$ denoted by $h^{\leftarrow} (\cdot)$ and $\nabla_x h(x)$ is derivative of $h$ with respect to $x.$ Also, we employ the notation, $\Phi_\alpha(x)=e^{-x^{-\alpha}},$ $x>0$ is the distribution of Fr\'{e}chet and $\Psi_\alpha(x)=e^{-\abs{x}^{\alpha}},$ $x<0$ is the distribution of Weibull with parameter $\alpha$. For right extremity of $F,$ we shall denote by $r(F)=\sup\{x:F(x)<1\}\leq\infty,$ and survival function is $\bar{F}(\cdot)=1-F(\cdot).$
 
\section{Distribution properties}

\subsection{Limiting distributions}
Throughout we consider measurable real valued function $U:\Real^+\to\Real^+$ is regularly varying function with index $\rho$ if 
$$\lim_{t\to\infty}\frac{U(tx)}{U(t)}=x^{\rho},\;\;\text{for }x>0.$$
We write $U\in RV_{\rho}$ and we call $\rho$ the exponent of variation.
The regular varying function plays an important role in the asymptotic analysis of various problems. It is well known, following de Haan and Ferreira (2006) that a necessary and sufficient condition for the existence of constants $a_n=a(n)$ and $b_n=\frac{1}{\bar{F}(n)}$ such that (\ref{Introduction_e1}) is equivalent
$$\lim_{t\to r(F)}\frac{\bar{F}(t+u(t)x)}{\bar{F}(t)}=(1+\xi x)_+^{-1/\xi},$$
 where, $u(t)=a(1/\bar{F}(t)).$
In this section, we establish the regular variation of the dfs belongs to the p-max domain of attraction of the log-GEV and negative log-GEV laws.
The following result reveals that the upper tail behavior of $F$ might determine whether $F\in D_p(L_{i,\xi}),$ $i=1,2.$ We first state and prove a lemma of independent interest which will be used subsequently.
\begin{lem}\label{lem}
If $\mathcal{L}$ is slowly varying as represented (\ref{rep}), then
\begin{eqnarray}
\mathcal{L}(e^y)=c(e^y)\exp\Big\{\int_{e^{x_0}}^{e^y}\frac{u^*(t)}{t\log t}dt\Big\},
\end{eqnarray}
where, $c(e^t)\to c$ and $u^*(t)=\log t\,u(t)\to 0$ as $t\to\infty.$
\end{lem}
\begin{proof}
From (\ref{rep}), taking $x=e^y,$ we have
\begin{eqnarray}
\mathcal{L}(e^y)&=&c(e^y)\exp\Big\{\int_{x_0}^{y}u(e^{t'})dt'\Big\},\;\;\;(\text{ where, } t=e^{t'}).\nonumber
\end{eqnarray}
Setting, $u^*(e^{t'})=t'\; u(e^{t'})\to 0,$ for lager $t',$ then
\begin{eqnarray}
\mathcal{L}(e^y)&=&c(e^y)\exp\Big\{\int_{x_0}^{y}\frac{u^*(e^{t'})}{t'}dt'\Big\},\nonumber\\
&=&c(e^y)\exp\Big\{\int_{e^{x_0}}^{e^y}\frac{u^*(t)}{t\log t}dt\Big\},\nonumber
\end{eqnarray}
where, $t'=\log t.$
\end{proof}
Now we obtain necessary and sufficient conditions for a df $F$ belongs to the p-max domain of attraction of  log-GEV and negative log-GEV stable laws. The next theorems examines the properties of regularly varying function for standardized these families.
\begin{thm}\label{thm1}
A df $F\in\D_p(L_{1,\xi}),$\\
(i) $r(F)=\infty,$ and $\xi>0$ if and only if 
\begin{eqnarray}
\lim_{t\to \infty}\dfrac{\bar{F}(x^{t\xi}e^t)}{\bar{F}(e^{t})}=(\log (x^{\xi}e))^{-\frac{1}{\xi}},\label{thm1_e1}
\end{eqnarray}
(ii) $0<r(F)<\infty,$ and $\xi<0$ if and only if 
\begin{eqnarray}
\lim_{t\to \infty}\dfrac{\bar{F}(r(F)e^{-t/\log(x^{\xi}e)})}{\bar{F}(r(F)e^{-t})}=(\log (x^{\xi}e))^{-\frac{1}{\xi}}.\label{thm1_e2}
\end{eqnarray}

\end{thm}
\begin{proof} (i) For $\xi> 0$ and $r(F)=\infty,$ we have $L_{1,\xi}(x) =K_{1,\alpha}(x^{\alpha^{-1}}e)$ for $\alpha=\frac{1}{\xi}.$ By Theorem \ref{RV1}-(a), is then equivalent to $\bar{F}(\exp(.))\in RV_{-\alpha}.$ Taking $z=\log (x^{\frac{1}{\alpha}}e)$ and from Theorem \ref{thm_von}-(1), setting $\frac{u(t)}{t\log t}=\frac{f(t)}{\bar{F}(t)}-\frac{\alpha}{t\log t}\to 0,$ for large $t$ and from Lemma \ref{lem},
\begin{eqnarray}
\dfrac{\mathcal{L}(e^{tz})}{\mathcal{L}(e^{t})}&=&\dfrac{c(e^{tz})\bar{F}(e^{t})}{c(e^{t})\bar{F} (e^{tz})}z^{-\alpha},\nonumber
\end{eqnarray}
Taking limit both side as $t\to\infty$ which is (\ref{thm1_e1}).  If (\ref{thm1_e1}) holds, choose $d_n =\log F^{\leftarrow} (1-1/n),$ then $1/\bar{F} (e^{d_n} )=n$ (see, Mohan and Ravi 1993) and then,
\[\lim_{n\to\infty}\dfrac{\bar{F}(x^{\frac{d_n}{\alpha}}e^{d_n})}{\bar{F}(e^{d_n})}=\lim_{n\to\infty}n\bar{F}(x^{\frac{d_n}{\alpha}}e^{d_n})=(\log (x^{\frac{1}{\alpha}}e))^{-\alpha},\]
whence, from (\ref{pmax_lim}), $F\in\D(L_{1,\xi}),$ for $\xi=\alpha^{-1}.$

(ii) Now, we have $L_{1,\xi}(x) =K_{2,\alpha}( x^{\frac{1}{\alpha}}e),$ for $\alpha=-\frac{1}{\xi},$ $\xi< 0$ and $0<r(F)<\infty.$ By Theorem \ref{RV2}-(a), $\bar{F}(r(F)\exp(-1/(.)))$ is regularly varying with exponent $(-\alpha).$ From Theorem \ref{thm_von}-(2), we choose $\frac{u(t)}{t\log (r(F)/t)}=\frac{f(t)}{\bar{F}(t)}-\frac{\alpha}{t\log (r(F)/t)}\to 0,$ for $t\to \infty$ and from Lemma \ref{lem},
\begin{eqnarray}
\dfrac{\mathcal{L}(r(F)e^{-t/z)})}{\mathcal{L}(r(F)e^{-t})}
&=&\dfrac{c(r(F)e^{-t/z})\bar{F}(r(F)e^{-t})}{c(r(F)e^{-t})\bar{F} (r(F)e^{-t/z})}z^{\alpha}.\nonumber
\end{eqnarray}
where, $z=\log (x^{-\frac{1}{\alpha}}e).$
Taking limit both side as $t\to \infty$ which is (\ref{thm1_e2}). Conversely, if (\ref{thm1_e2}) holds, setting $d_n =-\log\frac{r(F)}{F^{\leftarrow} (1-1/n)},$ then $1/\bar{F} (r(F)e^{-d_n})=n,$ (see, Mohan and Ravi 1993) and then,
\[\lim_{n\to\infty} \frac{\bar{F} (r(F)e^{-d_n/x})}{\bar{F} (r(F)e^{-d_n})} =\lim_{n\to\infty} n\bar{F} (r(F)e^{-d_n/x})=\left(\log (x^{-\frac{1}{\alpha}}e)\right)^{\alpha}.\]
whence again, from (\ref{pmax_lim}), $F\in\D(L_{1,\xi}),$
\end{proof}

\begin{thm}\label{thm3}
A df $F\in\D_p(L_{2,\xi}),$\\
(i) $r(F)=0,$ and $\xi>0$ if and only if 
\begin{eqnarray}
\lim_{t\to\infty}\dfrac{\bar{F}(-(-x)^{\xi t}e^{-t}))}{\bar{F}(-e^{-t})}=(-\log ((-x)^{\xi}e))^{-\frac{1}{\xi}}.\label{thm3_e1}
\end{eqnarray}
(ii) $r(F)<0,$ and $\xi<0$ if and only if 
\begin{eqnarray}
\lim_{t\to\infty}\dfrac{\bar{F}(r(F)e^{t/\log ((-x)^\xi e)})}{\bar{F}(r(F)e^{t})}=(-\log ((-x)^{\xi}e))^{-\frac{1}{\xi}};\label{thm4_e1}
\end{eqnarray}

\end{thm}
\begin{proof}
(i) We have $L_{2,\xi}(x) =K_{4,\alpha}(-(-x)^{\frac{1}{\alpha}}e)$ for $\alpha=\frac{1}{\xi}.$ Suppose, $\xi> 0$ and $r(F)=0,$  by Theorem \ref{RV1}-(b), $\bar{F}(-\exp(-(.)))\in RV_{-\alpha}.$ Putting $z=\log ((-x)^{-\frac{1}{\alpha}}e)$ and from Theorem \ref{thm_von}-(4), $\frac{u(t)}{t\log(-t)}=\frac{f(t)}{\bar{F}(t)}-\frac{\alpha}{t\log(-t)}\to 0,$ for $t\to\infty$ and from Lemma \ref{lem},
\begin{eqnarray}
\dfrac{\mathcal{L}(-e^{-tz})}{\mathcal{L}(-e^{-t})}=\dfrac{c(-e^{-tz})\bar{F}(-e^{-t})}{c(-e^{-t})\bar{F} (-e^{-tz})}z^{-\alpha},\nonumber
\end{eqnarray}
Taking limit both side as $t\to\infty$ and hence (\ref{thm3_e1}). Now, if (\ref{thm3_e1}) holds, define $d_n =-\log (-F^{\leftarrow} (1-1/n)),$ then $1/\bar{F} (-e^{-d_n} )=n$ (see, Mohan and Ravi 1993) and,
\[\lim_{n\to\infty}\dfrac{\bar{F}(-(-x)^{\frac{d_n}{\alpha}}e^{d_n})}{\bar{F}(-e^{-d_n})}=\lim_{n\to\infty}n\bar{F}((-x)^{-\frac{d_n}{\alpha}}e^{-d_n})=(\log ((-x)^{-\frac{1}{\alpha}}e))^{-\alpha}.\]
From (\ref{pmax_lim}), this implies that $F\in\D(L_{2,\xi}),$ for $\xi=\alpha^{-1}.$

(ii) Suppose, $\xi< 0$ and $r(F)<0,$ we have $L_{\xi}(x) =K_{5,\alpha}(-(-x)^{\frac{1}{\alpha}}e)$ for $\alpha=-\frac{1}{\xi}.$ By Theorem \ref{RV2}-(b), is then equivalent to $\bar{F}(r(F)\exp(1/(\cdot)))\in RV_{-\alpha}.$ From Theorem \ref{thm_von}-(5), we choose $\frac{u(t)}{t\log(r(F)/t)}=\frac{f(t)}{\bar{F}(t)}-\frac{\alpha}{t\log (r(F)/t)}\to 0,$ for $t\to \infty$ and from Lemma \ref{lem},
\begin{eqnarray}
\dfrac{\mathcal{L}(r(F)e^{t/z)})}{\mathcal{L}(r(F)e^{t})}
=\dfrac{c(r(F)e^{t/z})\bar{F}(r(F)e^{t})}{c(r(F)e^{t})\bar{F} (r(F)e^{t/z})}z^{\alpha}.\nonumber
\end{eqnarray}
where, $z=\log ((-x)^{-\frac{1}{\alpha}}e).$
Taking limit both side as $t\to r(F)$ which is (\ref{thm1_e1}).  If (\ref{thm1_e1}) holds, setting $d_n =\log\frac{r(F)}{F^{\leftarrow} (1-1/n)},$ then $1/\bar{F} (r(F)e^{d_n} )=n$ (see, Mohan and Ravi 1993) and then,
\[\lim_{n\to\infty} \frac{\bar{F} (r(F)e^{d_n/x})}{\bar{F} (r(F)e^{d_n})} =\lim_{n\to\infty} n\bar{F} (r(F)e^{d_n/x})=\left(\log((-x)^{\frac{1}{\alpha}}e)\right)^{\alpha}.\]
whence, from (\ref{pmax_lim}), $F\in\D(L_{2,\xi}),$ for $\xi=-\alpha^{-1}.$
\end{proof}
\begin{rem} In case of $\xi=0$ (eq. \ref{gev}) for  $L_{\tilde{\xi}}(x)=\Phi_1(x)$ is proved in Theorem \ref{Ap_thm1} and for $L_{\tilde{\xi}}(x)=\Psi_1(x)$ presented Theorem \ref{Ap_thm2}.
\end{rem}

\subsection{Moments}
Some of the most important features and characteristics of a distribution can be studied through moments. The $k$th moments of PGEV are derived in the following theorems. In our proofs of $k$th moments of PGEV, the moment generating function (MGF) of Weibull with positive support plays and important role. Note that the MGF corresponding to a standard Weibull rv of $Y$ with positive support specified as 
\begin{eqnarray}\label{mgf1}
M_Y(t;\alpha)=\alpha\int_{0}^{\infty}x^{\alpha-1}\exp\left(-tx-x^{\alpha}\right)dx.\label{mgf_weibull}
\end{eqnarray}
Cheng et al. (2004) derived the moment generating function (MGF) of $Y,$ when the parameter $\alpha$ takes integer values. Nadarajah and Kotz (2007) show that a closed form expression for MGF of $Y,$ for all rational values of shape parameter. Since, we assume $\alpha=p/q,$ where $p\geq 1$ and $q\geq 1,$ are coprime integers, the integral in (\ref{mgf_weibull}) can be  provided that
\begin{eqnarray}\label{mgf_1}
M_Y(t;\alpha)=\left\lbrace
\begin{array}{ll}
\alpha\sum\limits_{j=0}^{q-1}\frac{(-1)^j}{j! t^{\alpha+\alpha j}}\Gamma(\alpha+\alpha j) \left[{_{p+1}\mathcal{F}_q}(1,\Delta(p,j\alpha+j);\Delta(q,1+j);(-1)^q z)\right],&\\
\text{if }0<\alpha<1;&\\
&\\
\sum\limits_{j=0}^{p-1}\frac{(-t)^j}{j!}\Gamma\left(1+\frac{j}{\alpha}\right)\left[_{q+1}\mathcal{F}_p\left(1,\Delta\left(q,1+\frac{\alpha}{j}\right);\Delta(p,1+j);\frac{(-1)^p}{z}\right)\right],&\\
\text{ if }\alpha>1,&\\
\end{array}\right.
\end{eqnarray}
where, $z=p^p/(t^p q^q)$ and $\Delta(c,d)=\{d/c,(d+1)/c,\cdots,(c+d-1)/c\}$ and ${_p\mathcal{F}_q}$ is the generalized hyper geometric function defined by $${_p\mathcal{F}_q}(a_1,\cdots,a_p;b_1,\cdots,b_q;x)=\sum_{k=0}^{\infty}\frac{(a_1)_k(a_2)_k\cdots (a_p)_k}{(b_1)_k(b_2)_k\cdots (b_q)_k}\frac{x^k}{k!}$$
where, $(\upsilon)_k=\upsilon(\upsilon+1)\cdots(\upsilon+k-1).$
 In particular value $\alpha=1$ simple integration of (\ref{mgf_weibull}) gives,
 \begin{eqnarray}\label{mgf_2}
 M_Y(t;1)=\frac{1}{1+t}.
 \end{eqnarray}
 In the case $\alpha=2$ the MGF becomes
 \begin{eqnarray}\label{mgf3}
 M_Y(t;2)=1-\frac{t\sqrt{\pi}}{2}\exp\left(\frac{t^2}{4}\right)\erf\left(\frac{t}{2}\right),
 \end{eqnarray}
where, the complementary error function defined by $\erf(x)=1-\frac{2}{\sqrt{\pi}}\int_{0}^{x}\exp(-t^2)dt.$
The generalized hypergeometric function is widely available in many scientific software packages, such as R and Matlab.

The following results show that, the proofs of the $k$th moments of PGEV involve the application of MGF of standard Weibull distribution function.
\begin{thm}
Let $Y$ is a rv with standard Weibull df and $X$ is a rv with PGEV in (\ref{df3}). For $k>0,$ 
\begin{enumerate}
	\item[(i)] $X^+$ is positive support and $\xi<0$
	\begin{eqnarray}
E(X^+)^k=e^{k\left(\mu-\frac{\sigma}{\xi}\right)}M_Y\left(\frac{k\sigma}{\abs{\xi}},\frac{1}{\abs{\xi}}\right).\nonumber
	\end{eqnarray}
where, $M_Y(\cdot)$ defined in (\ref{mgf_1}).

\item[(ii)] $X^-$ is negative support and $\xi>0$
\begin{eqnarray}
E\abs{X^-}^k=e^{k\left(\mu+\frac{\sigma}{\xi}\right)}M_{Y^{-1}}\left(\frac{k\sigma}{\xi},\frac{1}{\xi}\right).\nonumber
\end{eqnarray}
\end{enumerate}
\end{thm}
\begin{proof} Suppose $\abs{Z}=\left(\abs{X}e^{-\mu}\right)^{\frac{1}{\sigma}}$ is a standardizing rv with df in (\ref{df3}) for $A=\{z:1+\text{sign}(z)\xi\log \abs{z}>0\}.$ We write
\begin{equation}
E\abs{Z}^k=\int_{A} \abs{z}^{k-1}(1+\text{sign}(Z)\xi\log \abs{z})^{-1-1/\xi}e^{-(1+\text{sign}(Z)\xi\log \abs{z})^{-1/\xi}}dz,
\end{equation}
We have
\begin{equation}\label{moment}
E\abs{Z}^k=\frac{1}{\xi}\int_{A} e^{\frac{k}{\xi}(y-1)\text{sign}(Z)-y^{-\frac{1}{\xi}}}y^{-1-\frac{1}{\xi}}dy,\;\text{ where, } y=1+\xi\,\text{sign}(Z)\log\abs{z}.
\end{equation}
(i) Let $Z^+$ is a rv with positive support. From (\ref{moment}), the $k$th moment does not exist for $\xi>0.$ For $\xi<0,$ we have
\begin{eqnarray}
E(Z^+)^k&=&-\frac{1}{\xi}\int_{0}^{\infty} e^{\frac{k}{\xi}(y-1)-y^{-\frac{1}{\xi}}}y^{-1-\frac{1}{\xi}}dy,\nonumber\\
&=&e^{-\frac{k}{\xi}}M_Y\left(\frac{k}{\abs{\xi}},\frac{1}{\abs{\xi}}\right).\nonumber
\end{eqnarray}
where, $Y$ is a positive rv with standard Weibull distribution and
$M_Y(\cdot)$ defined in (\ref{mgf_1}).
The $k$th moment of $X^+$ can be obtained as
\begin{eqnarray}
E(X^+)^k=e^{k\left(\mu-\frac{\sigma}{\xi}\right)}M_Y\left(\frac{k\sigma}{\abs{\xi}},\frac{1}{\abs{\xi}}\right).\nonumber
\end{eqnarray}
(ii) Similarly, let $Z^-$ is a rv with neagitve support. From (\ref{moment}), the $k$th moment does not exist, for $\xi<0.$ For $\xi>0$ we get
\begin{eqnarray}
E\abs{Z^-}^k&=&\frac{1}{\xi}\int_{0}^{\infty} e^{\frac{k}{\xi}(1-y)-y^{-\frac{1}{\xi}}}y^{-1-\frac{1}{\xi}}dy,\nonumber\\
&=&e^{\frac{k}{\xi}}M_{Y^{-1}}\left(\frac{k}{\xi},\frac{1}{\xi}\right).\nonumber
\end{eqnarray}
The $k$th moment of $X^-$ can be obtained as
\begin{eqnarray}
E\abs{X^-}^k=e^{k\left(\mu+\frac{\sigma}{\xi}\right)}M_{Y^{-1}}\left(\frac{k\sigma}{\xi},\frac{1}{\xi}\right).\nonumber
\end{eqnarray}
\end{proof}
\begin{rem}
The $k$th moment of rvs $X^+$ with PGEV for $\xi>0$ and the $k$th moment of rvs $X^-$ with PGEV $\xi<0$ do not exist.
\end{rem}
The $k$th central moments of $X$ are easily obtained from the ordinary moments by 
\begin{eqnarray}\label{cm}
E(X-E(X))^k=\sum_{j=0}^{k}{{k}\choose{j}}(-1)^j(E(X))^jE(X^{k-j}).
\end{eqnarray}
From (\ref{cm}) and $k=2,$ the variances of PGEV listed in Appendix \ref{tab1}.

\subsection{Entropy.}
An entropy of rv $X$ is a measure of variation of the uncertainty. Shannon entropy is defined by
\begin{eqnarray}\label{Shannon_entropy}
H(X)=-\int_A \log f(x) f(x)dx,
\end{eqnarray}
where, $A=\{x: f(x)>0\}.$ 
 Here, the Shannon entropy of GEV family is well known as
\begin{equation}
H(X)=\log\tilde{\sigma}+(\tilde{\xi}+1)\gamma+1.
\end{equation} 
The Shannon entropy of six type of p-max stable laws are evaluated by Ravi and Saeb (2012). 
Now, we illustrate the Shannon entropy of PGEV family.
\begin{thm}\label{thm_ent}
If $X$ is a rv with df PGEV for $\xi< 0,$ then the Shannon entropy of $X$ is given by
\begin{eqnarray}
H(X)=
\mu+\log\sigma+(1+\xi)\gamma+\frac{\sigma}{\xi}E(\text{sign}(X))\left[\Gamma(1-\xi)-1\right]+1.
\end{eqnarray}
\end{thm}
\begin{proof}
Let $Z$ is a standardized rv with df PGEV ($\xi<0$) in (\ref{pdf_1}), the Shannon entropy is given by
\begin{eqnarray}
H(Z)&=&E(\log \abs{Z})+E\left(\log (1+\text{sign}(Z)\xi\log \abs{Z})^{1+1/\xi}+(1+\text{sign}(Z)\xi\log \abs{Z})^{-1/\xi}\right),\nonumber\\
&=&E_1+E_2.\label{ZEnt}
\end{eqnarray}
Putting $Y=(1+\text{sign}(Z)\xi\log \abs{Z})^{-1/\xi},$ and $Y$ has standard exponential distribution. 

\begin{eqnarray}
E_1=\xi^{-1}E(\text{sign}(Z))E(Y^{-\xi}-1))
=\frac{1}{\xi}E(\text{sign}(Z))\left[\Gamma(1-\xi)-1\right].\label{E1}
\end{eqnarray}
Next, 
\begin{eqnarray}
E_2=-(1+\xi)E_Y(\log(Y))+E_Y(Y)
=(1+\xi)\gamma+1,\label{E2}
\end{eqnarray}
From (\ref{E1}), (\ref{E2}) we get
\begin{eqnarray}
H(Z)=
(1+\xi)\gamma+\frac{1}{\xi}E(\text{sign}(Z))\left[\Gamma(1-\xi)-1\right]+1.
\end{eqnarray}
From lemma 1.3, Ravi and Saeb (2012), If $X=\abs{Z}^{\sigma} e^{\mu}$ then
\begin{eqnarray}
H(X)&=&\mu+\log\sigma+\left(\sigma-1\right)E\log \abs{Z}+H(Z), \nonumber\\
&=&\mu+\log\sigma+(1+\xi)\gamma+\frac{\sigma}{\xi}E(\text{sign}(X))\left[\Gamma(1-\xi)-1\right]+1.\nonumber
\end{eqnarray}
\end{proof}

\begin{rem}\label{rem2}
Note that, the Shannon entropy of the PGEV distribution for $\xi>0$ does not exist.
\end{rem}
 Suppose $Y$ is a rv with df $F_Y$ and $X=h(Y)$ with df $F_X$ where $h$ is a continuous function. The entropy ordering $H(Y)<H(X),$ will be denoted as $F_Y\stackrel{E}{<}F_X$ or $Y\stackrel{E}{<}X.$ In general case, the following lemma finds a direct relationship for entropy.
\begin{lem}\label{lem2}
 If $E_X\left(\log\abs{\nabla_X h^{\leftarrow}(X)}\right)<0$ then $Y\stackrel{E}{<}X.$
\end{lem}
\begin{proof}
We write,
\[F_X(x)=\Pr(h(Y)\leq x)=\Pr(Y\leq h^{\leftarrow}(x))=F_Y(h^{\leftarrow}(x)),\]
with respective density function
\[f_X(x)=f_Y(h^{\leftarrow}(x))\abs{\nabla_x(h^{\leftarrow}(x))}.\]
From definition of entropy we have
\begin{eqnarray}\label{lem1}
H(X)&=&-\int_\Real f_Y(h^{\leftarrow}(x))\log\left(f_Y(h^{\leftarrow}(x)\right) d(h^{\leftarrow}(x))-\int_{\Real}f_X(x)\log\abs{\nabla_x h^{\leftarrow}(x)}dx,\nonumber\\
&=&H(Y)-E_X(\log\abs{\nabla_X(h^{\leftarrow}(X))}).\label{entropy}
\end{eqnarray}
Noting that, if $E_X\left(\log\abs{\nabla_X h^{\leftarrow}(X)}\right)<0$ then $Y\stackrel{E}{<}X.$
\end{proof}
The following theorem investigates the entropy ordering in GEV families with $\xi<0.$
\begin{thm}\label{entropy_ord} Suppose $Y$ has GEV family.
If $X=\text{sign(X)}\exp(\abs{Y})$ is a rv with PGEV $(\xi<0)$ then $Y\stackrel{E}{<}X.$ 
\end{thm}
\begin{proof}
(i) Let $X^{+}$ is a positive rv and $h(x)=\exp(x).$ From Lemma \ref{lem2}, it is enough to show that, $E_X(\log(X))>0.$ Use the proof of Theorem \ref{thm_ent}, we have
$E_X(\log(X))=\mu+\frac{\sigma}{\xi}(\Gamma(1+\abs{\xi})-1).$ Since, the Shannon entropy of PGEV for $\xi<0$ exists, $\Gamma(1+\abs{\xi})>0$ for all $\xi<0,$ and $Y\stackrel{E}{<}X^{+}$ holds.

(ii) Similarly, define $h(x)=-\exp(-x)$ and $X^{-}$ is a negative rv. From Lemma \ref{lem2} and Theorem \ref{thm_ent}, $E_X\left(\log\abs{X}\right)=\mu-\frac{\sigma}{\xi}(\Gamma(1-\xi)-1)>0$ for all $\xi<0$ and hence the proof.
\end{proof}
  
\subsection{Dispersion ordering.} 
Lewis and Thompson (1981) have defined the concept of \textquotedblleft ordering in dispersion\textquotedblright. Two distributions $F_X$ and $F_Y$ are said to be ordered in dispersion, denoted by $F_Y\stackrel{disp}{<}F_X$ if and only if
\[F_Y^{\leftarrow}(u)-F_Y^{\leftarrow}(v)\leq F_X^{\leftarrow}(u)-F_X^{\leftarrow}(v),\;\;\text{ for all }0<v<u<1.\]
It is easily seen by putting $u=F_Y(y)$ and $v=F_Y(x)$ where $y\leq x$ that $F_Y\stackrel{disp}{<}F_X$ if and only if
\begin{eqnarray}\label{dis_1}
F_X^{\leftarrow}(F_Y(x))+x\;\text{ is nondecreasing in }x,
\end{eqnarray}
then, $F_Y$ is said to be tail-ordered with respect to $F_X$ $(F_Y \stackrel{tail}{<} F_X).$ Thus we see that dispersive ordering is the same as tail-ordering. Oja (1981) shows that the dispersion ordering implies both variance ordering and entropy ordering $(\stackrel{EV}{<}).$ In other word, $F_Y\stackrel{disp}{<}F_X$ is a sufficient condition for $Y\stackrel{EV}{<}X$ (variance and entropy order similarly). Entropy ordering of distributions within many parametric families are studied in Ebrahimi et al. (1999).

Let  $L(x_p)=1-p,$ where, $L(\cdot)$ is the distribution (\ref{df3}) so that
\begin{eqnarray}\label{quantile_e1}
x_p=\text{sign}(X)\exp\left(\frac{\sigma}{\xi}\text{sign}(X)\left(y_p^{-\xi}-1\right)+\mu\right);
\end{eqnarray}
We also well known the quantile for $\xi=0$ in (\ref{gev}) we get
\begin{eqnarray}
x_p=\frac{\tilde{\sigma}}{\tilde{\xi}}(y_p^{-\tilde{\xi}}-1)+\tilde{\mu},\label{quantile_e2}
\end{eqnarray}
where $y_p =-\log(1-p).$
 The following corollary investigates the dispersion ordering in the GEV families. 
\begin{coll} Suppose $X$ and $Y$ are rvs to correspond PGEV and GEV families. Let $X^+$ is a positive support, from (\ref{dis_1}) and (\ref{quantile_e1}) we have
\begin{eqnarray}\label{dispresion}
L^{\leftarrow}(G(x))=\exp(x)+x,\nonumber
\end{eqnarray}
is a nondecreasing function for all $x$ in support of GEV, then, $Y\stackrel{disp}{<}X^+.$ On the other hand, the result from Oja (1981) and Theorem \ref{entropy_ord}, the entropy of GEV and PGEV families are ordered for $\xi<0,$ we conclude that the variances are also ordered in $\xi,$
so, $Y\stackrel{EV}{<}X^+$ for $\xi<0.$
Similarly, if $X^-$ is a rv with negative support, from Theorem \ref{entropy_ord}, $Y\stackrel{EV}{<}X^-$ for $\xi<0$ and hence the proof.
\end{coll}  

\section{Methods of Estimations}
\subsection{Maximum Likelihood Estimation.} The method of maximum likelihood estimation (MLE) using Newton-Raphson iteration to maximize the likelihood function of GEV, as recommended by Prescott and Walden (1980).
The log-likelihood function  for $(\mu,\sigma,\xi)$ based on PGEV family, given by
\begin{eqnarray}\label{mle2}
\ell(x;\mu,\sigma,\xi)&=&-k\log \sigma-\sum_{i=1}^{k}\log \abs{x_i}-\left(1+\frac{1}{\xi}\right)\sum_{i=1}^{k}\log\left(1+\xi\,\text{sign}(x)\left(\frac{\log \abs{x_i}-\mu}{\sigma}\right)\right)\nonumber\\
&&-\sum_{i=1}^{k}\left(1+\xi\,\text{sign}(x)\left(\frac{\log\abs{ x_i}-\mu}{\sigma}\right)\right)^{-1/\xi};
\end{eqnarray}
For determining the MLEs of the parameters $\mu,$ $\sigma$ and $\xi,$ we can use the same procedure as for the GEV law. Since, there is no analytical solution, but for any given dataset the maximization is straightforward using standard numerical optimization algorithms. Jenkinson (1969), Prescott and Walden (1980) show that the elements of the Fisher information matrix for GEV distribution$(\xi\neq 0).$
Since the $\log \abs{x}$ is free from of parameters, 
the Fisher information matrix for PGEV is similar the Fisher information matrix for GEV law. Since, the Shannon entropy is equivalent to the negative log-likelihood function and from remark \ref{rem2} the MLEs exists for $\xi<1.$
Smith (1985) has investigated the classical regularity conditions for the asymptotic properties of MLEs are not satisfied but he shows that, when $\xi>-0.5$ the MLEs have usual asymptotic properties. For $\xi=-0.5$ the MLEs are asymptotically efficient and normally distributed, but with a different rate of convergence. We remark that results of Smith applies also to the three parameters. The MLEs may nonregular for $\xi< -0.5$ and $\xi\geq 1$, but Bayesian techniques offer an alternative that is often preferable. 
\subsection{Bayesian Estimation.} Let $\bm{\theta}$ is a vector of the model parameters in a space  $\Theta$ and $\pi(\bm{\theta})$ denote the density of the prior distribution for $\bm{\theta}.$ The posterior density of $\bm{\theta}$ is given by
\begin{eqnarray}
\pi(\bm{\theta}|\bm{x})=\frac{\pi(\bm{\theta})\exp(\ell (\bm{x};\bm{\theta}))}{\int_{\Theta} \pi(\bm{\theta})\exp(\ell (\bm{x};\bm{\theta}))d\bm{\theta}}\propto \pi(\bm{\theta})\exp(\ell (\bm{x};\bm{\theta})).\nonumber
\end{eqnarray}
where, $\ell(\cdot)$ is log-likelihood function. However, computing posterior inference directly is difficult. To bypass this problem we can use simulation bases techniques such as Markov Chain Monte Carlo (MCMC). 
The Markov Chain is generated using standard Metropolis (Hastings, 1970) within Gibbs (Geman and Geman, 1984) methods. Setting $\bm{\theta}=(\mu,\eta, \xi)$ where, $\eta=\log\sigma$ is easier to work. 
We might choose a prior density function
\[\pi(\bm{\theta})=\pi_\mu(\mu)\pi_\eta(\eta)\pi_\xi(\xi),\]
where the marginal priors, $\pi_\mu(\cdot), \pi_\eta(\cdot)$ and $\pi_\xi(\cdot),$ are normal density function with mean zero and variances, $v_\mu,\;v_\eta$ and $v_\xi$ respectively.
These are independent normal priors with large variances. The variances are chosen large enough to make the distributions almost flat and therefore should correspond to prior ignorance. The choice of normal densities is arbitrary. The proposed value $\bm{\theta}^*$ at point $i$ is $\bm{\theta}^* = \bm{\theta}^{(i)} + \bm{\epsilon}.$ The
$\bm{\epsilon}=(\epsilon_\mu, \epsilon_\eta, \epsilon_\xi)$ are normally distributed variables, with zero means and variances $w_\mu,$ $w_\eta$ and $w_\xi$ respectively. 

Now we specify an arbitrary probability rule $q(\bm{\theta}_{i+1}|\bm{\theta}_i)$ for iterative simulation of successive values. The distribution $q$ is called the proposal distribution. Possibilities include $(\bm{\theta}_{i+1}| \bm{\theta}_i)$ is Normal density with mean $\bm{\theta}_i$ and variance one. Then $q(\bm{\theta}_i| \bm{\theta}^*) =\tilde{f}(\bm{\theta}^*-\bm{\theta}_i),$ where $\tilde{f}(\cdot)$ is
the density function of $\bm{\epsilon}.$ Since the distribution of $\bm{\epsilon}$ is symmetric about zero $q(\bm{\theta}_i|\bm{\theta^*})=q(\bm{\theta^*}|\bm{\theta}_i).$
The acceptance probability
\begin{eqnarray}\label{Omega}
\Omega_i=\min\left\{1,\frac{\exp(\ell(x;\bm{\theta}^*))\pi(\bm{\theta}^*)}{\exp(\ell(x;\bm{\theta}_i))\pi(\bm{\theta}_i)}\right\},
\end{eqnarray}
was suggested by Hastings (1970). Here we accepted the proposed value $\bm{\theta}^*$ with probability $\Omega_i.$
We note that, the variance of $\bm{\epsilon}$ affects the acceptance
probability, if the variance is too low most proposals will be accepted, resulting in very slow convergence, and if it is too high very few will be accepted and the moves in the chain will often be large. Appendix \ref{Ape1} gives the details of the required algorithm.

Here we find few papers linking the Bayesian method and extreme value analysis. Smith and Naylor (1987) who compare Bayesian and maximum likelihood estimators for the Weibull distribution. Coles and Powell (1996) and Coles and Tawn (1996) for a
detailed review of Bayesian methods in extreme value modelling. 
Stephenson and Tawn (2004) perform inference using reversible jump MCMC techniques for extremal types.
\subsection{Prediction.} We are interested
in the outcome $y$ of the future experiment.
Within the Bayesian framework, the predictive distribution  function is argued by Aitchison and Dunsmore (1975). In particular, since
the objective of an extreme value analysis is usually an estimate of the
probability of future events reaching extreme levels, expression through
predictive distributions is natural. Let $Y$ is a rv with annual maximum distribution over a
future period of years and $\bm{x}$  represents historical observations. The predictive distribution  function is defined as
\begin{eqnarray}\label{pred}
\Pr(Y< y|\bm{x})&=&\int_\Theta \Pr(Y<y|\bm{\theta})\pi(\bm{\theta}|\bm{x})d\bm{\theta},\nonumber\\
&\simeq&\frac{1}{n}\sum_{i=1}^{n}\Pr(Y < y|\bm{\theta}_i),\nonumber
\end{eqnarray}
where $\bm{\theta}_i$ is the output from the $i^{th}$ iteration of a sample of size $n$ from the Gibbs
sampler of posterior distribution of $\bm{\theta}.$ Estimates of extreme quantiles
of the annual maximum distribution are then obtained by solving the equation
\begin{eqnarray}\label{return}
\frac{1}{n}\sum_{i=1}^{n}\Pr(Y < x_p|\bm{\theta}_i)=1-p,
\end{eqnarray}
for $x_p$ with various values of $m$ where, $m=1/p$ is defined as return period.
\subsection{Real Data Analysis.}
In this section we shall use the PGEV model to a real data set. This analysis is based on the annual maximum yearly rainfall data of station Eudunda, Australia (Latitude 34.18S; Longitude 139.08E and
Elevation 415 m) which collected during 1881-2015. Annual maxima, corresponding to the year from 1881, were found from the 135 years worth of data and are plotted in Fig \ref{Data}. 
We assume that the pattern of variation has stayed constant over the observation period, so we model the data as independent observations from the GEV families. 
Here, maximization of GEV and PGEV log-likelihood functions using the "Nelder-Mead" algorithms. All the computations were done using R programming language.

In what follows we shall apply formal goodness of fit tests in order to verify which distribution fits better to these data. We apply the Cram\'{e}r-von Mises ($C$) and Anderson-Darling ($A$) test statistics. The test statistics $C$ and $A$ are described in detail in Chen and Balakrishnan (1995). In general, the smaller values of statistics $C$ and $A$, the better fit to data. Additionally, from the critical values of statistics $C$ and $A$ given in Chen and Balakrishnan (1995), it is possible to calculate the p-values corresponding to each test statistics. The null hypothesis is $H_0:\{X_1,\ldots,X_n\}$ comes from GEV/PGEV families. To test $H_0,$ we can proceed as appendix \ref{Ape}.
The values of statistics $C$ and $A$ (p-values between parentheses) for all models are given in Table \ref{Table1}.
From this table we conclude that does not evidence to reject the null hypothesis for GEV/PGEV distributions.
Table \ref{Table2} lists the MLE method of the parameters estimation and standard errors in parentheses. Since the values of standard errors in PGEV model are lower than other laws, we suggesting that the PGEV model is best fit model for these data. 
Within the Bayesian model with non-informative prior distributions, the algorithm in \ref{Ape1} was applied to annual maxima dataset. Initializing the MCMC algorithm with maximum likelihood estimates as our initial vector, $\theta_0 = (4.3614, 02853, -02386)$ should produce a chain with small burn-in period. After some pilot runs, a Markov chain of $1000$
iterations was then generated with good mixing properties (Figure \ref{graph_bayes}). The burn-in period was taken to be the first $400$ iterations which the stochastic variations in the chain seem reasonably homogeneous. If we accept this, after deleting the
first $600$ simulations, the remaining $400$ simulated values can be treated as
dependent realizations whose marginal distribution is the target posterior. The sample means (and standard error) of each marginal component
of the chain are
\[\hat{\mu}=4.3615\;(0.0265)\;\;\hat{\sigma}= 0.2848\;(0.0144)\;\;\hat{\xi}= -0.2411\;(0.0340).\]
Finally,  using eq.  \ref{return}, a plot of the predictive distribution of a future annual maximum is shown in Fig. \ref{return_graph} on the usual return period scale. Table.\ref{Table3} shows the predictive return levels $x_p$ for $m$ years where, $m=\frac{1}{p}$ is return period. For example, the corresponding estimate for the 4 years return level is $x_{0.75}= 106.59.$


\appendix
\section{}\label{pmax}
The p-max stable laws, namely,
\begin{equation*}
\begin{array}{ll}
\text{the log-Fr\'{e}chet law:}  &
 K_{1,\alpha}(x) =  \left\lbrace	
           \begin{array}{l l}
			0, & x<1,\\
			\exp(-(\log x)^{-\alpha}), & 1\leq x;
			\end{array}
			\right.\\ 
\\
\text{the log-Weibull law:} &
 K_{2,\alpha}(x)=  \left\lbrace
		\begin{array}{l l}
			0, & x<0, \\
			\exp(-(-\log x)^{\alpha}), & 0\leq x<1, \\
			1, & 1\leq x;
			\end{array}\right.\\ 
\\
\text{the standard Fr\'{e}chet law: } &
K_{3}(x)  =  \Phi_1(x), x \in \Real; 
\\
\text{the negative log-Fr\'{e}chet law: }&	
K_{4,\alpha}(x)  = \left\lbrace
		\begin{array}{l l}
			0, & x<-1, \\
			\exp(-(-\log (-x))^{-\alpha}), & -1\leq x<0,\\
			1, & 0\leq x;
			\end{array}\right. \\
\\
\text{the negative log-Weibull law: } &	K_{5,\alpha}(x)  =  \left\lbrace
		\begin{array}{l l}
		\exp(-(\log (-x))^{\alpha}) & x<-1,\\
			1, & -1\leq x; 
			\end{array}\right.\\
\\
\text{the standard Weibull law:}& K_6(x) = \Psi_1(x), x \in \Real;
\end{array}
\end{equation*}
where, $\alpha>0$ being a parameter.

\section{}\label{fig}
The density function of (\ref{df3}), respectively, given by
\begin{eqnarray}\label{pdf_1}
l_{\xi}(x;\mu,\sigma)&=&
\frac{1}{\sigma\abs{x}}\exp\left(-\left(1+\frac{\xi}{\sigma}\,\text{sign}(x)\left(\log\abs{x}-\mu\right)\right)_+^{-1/\xi}\right)\nonumber\\
&&\left(1+\frac{\xi}{\sigma}\,\text{sign}(x)(\log \abs{x}-\mu)\right)_+^{-1-1/\xi};
\end{eqnarray}
And from (\ref{gev}), density function of PGEV distribution with $\xi=0$ is well known as
\begin{eqnarray}
l_{\tilde{\xi}}(x;\tilde{\sigma})=\frac{1}{\tilde{\sigma}}\exp\left(-\left(1+\frac{\tilde{\xi}(x-\tilde{\mu})}{\tilde{\sigma}}\right)^{-1/\tilde{\xi}}\right)\left(1+\frac{\tilde{\xi}(x-\tilde{\mu})}{\tilde{\sigma}}\right)^{-1-1/\tilde{\xi}}.\label{pdf2}
\end{eqnarray}

A quantile estimator and variance of $x_p$ are defined by substituting estimators $\mu,$ $\sigma$ and $\xi$ for the parameters in (\ref{quantile_e1}) and (\ref{quantile_e2}). Note that $x_p$ is a function of $\mu, \sigma$ and $\xi$ and it is a rv.
The variance of $x_p$ is given by the delta method,
\begin{eqnarray}\label{v_MLE}
Var(x_p)=\nabla_{\bm\theta} x_p^{T}\bm\Sigma\nabla_{\bm\theta} x_p,
\end{eqnarray}
where, 
$\bm\Sigma$ is variance covariance matrix 
and $\nabla_{\bm\theta}x_p^T$ for $\bm{\theta}=[\mu,\sigma,\xi]$ is calculating by
\begin{eqnarray}
\nabla_{\bm\theta} x_p^{T}=	\abs{x_p}\left[1,\xi^{-1}(y_p)^{-\xi}-1),\frac{\sigma}{\xi^{2}}\left[(y_p)^{-\xi}\log(y_p)^{-\xi}-\left((y_p)^{-\xi}-1\right)\right]\right].\nonumber
\end{eqnarray}
where, $\xi\neq 0.$ For $\xi= 0$  (\ref{quantile_e2}) is still valid for $\bm{\tilde{\theta}}=[\tilde{\mu},\tilde\sigma,\tilde\xi]$ with
\begin{eqnarray}
\nabla_{\bm{\tilde{\theta}}} x_p^{T}=\left[1,-\tilde{\xi}^{-1}(1-(y_p)^{-\tilde{\xi}}),\tilde{\sigma}\tilde{\xi}^{-2}(1-(y_p)^{-\tilde{\xi}})-\frac{\tilde{\sigma}}{\tilde{\xi}}(y_p)^{-\tilde{\xi}}\log(y_p)\right].\nonumber
\end{eqnarray}
where $y_p =-\log(1-p).$ Approximate confidence intervals (CI) can also be obtained by the delta method. The delta method enable the approximate normality of $\hat{x}_p$ to be used to obtain CI for $x_p.$ It follows that an approximate $(1-\alpha)$ CI for $x_p$ is $\hat{x}_p\pm z_{\alpha/2}\sqrt{Var(\hat{x}_p)}.$

In following, illustrate figures \ref{fig2} and \ref{fig3} of standardized density function $l_{\xi}$ for different values of $\xi.$
\begin{figure}[!ht]
	\centering
	\includegraphics[width=.7\textwidth]{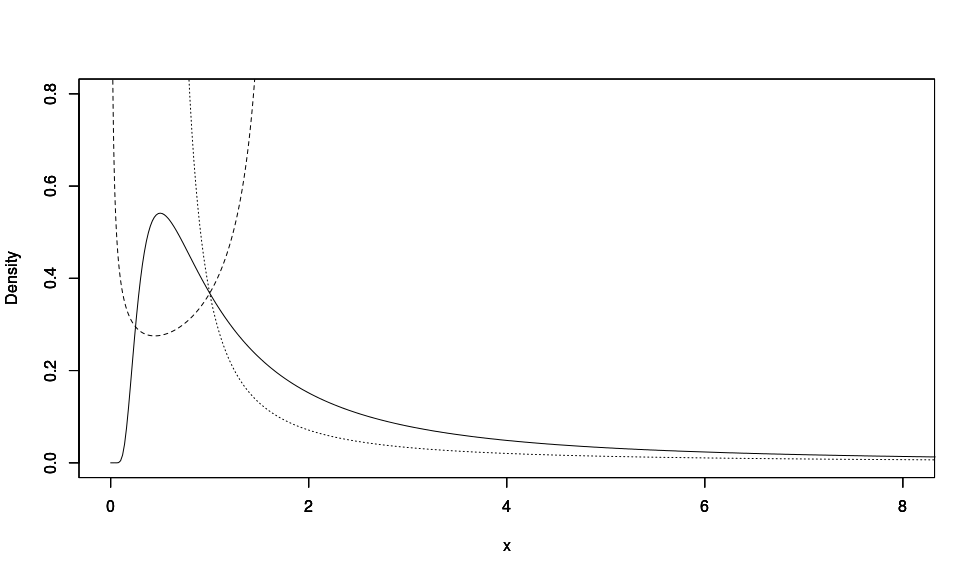}
	\caption{graph of density function  with positive support for $\xi=-2$ (dash), $\xi=2$ (dots) and standard Fr\'{e}chet (line).}\label{fig2}
\end{figure}
\begin{figure}[!ht]
	\centering
	\includegraphics[width=.7\textwidth]{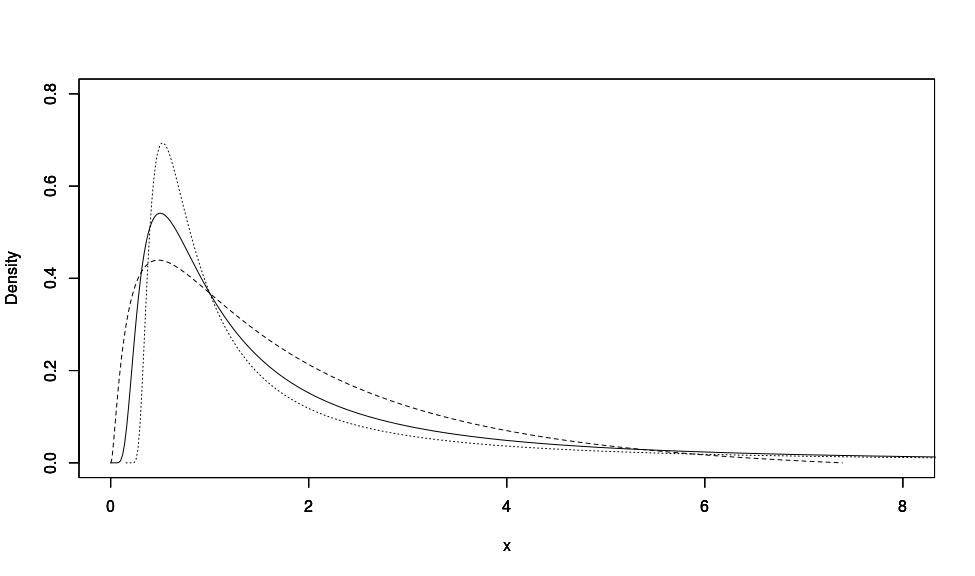}
	\caption{graph of density function  with positive support for $\xi=-0.5$ (dash), $\xi=0.5$ (dots) and standard Fr\'{e}chet (line)}\label{fig3}
\end{figure}

\newpage
\section[Appendix B]{}
\begin{thm}(Resnick 1987, The Karamata representation)
$\mathcal{L}$ is slowly varying iff $\mathcal{L}$ can be represented as
\begin{eqnarray}
\mathcal{L}(x) = c(x) \exp{\Big\{\int_{z}^{x}\frac{u(t)}{t}dt\Big\}},\;\;z<x<\infty,\label{rep}
\end{eqnarray}
where, $c(t)\to c > 0$ and $u(t)\to 0$ as $t \to \infty$ locally uniformly. 
\end{thm}

\begin{thm}\label{RV1}(Mohan and Ravi 1993)
(a) $F\in\D_p(K_{1,\alpha})$ iff $r(F)=\infty$ and $\bar{F}(\exp(.))$ is regularly varying at $\infty$ with exponent $(-\alpha).$
(b) $F\in\D_p(K_{4,\alpha})$ iff $r(F)=0$ and $\bar{F}(-\exp(-(.)))$ is regularly varying at $\infty$ with exponent $(-\alpha).$
\end{thm}
\begin{thm}\label{RV2}(Mohan and Ravi 1993)
(a) $F\in\D_p(K_{2,\alpha})$ iff $0<r(F)<\infty$ and $\bar{F}(r(F)\exp(-1/(.)))$ is regularly varying at $\infty$ with exponent $(-\alpha).$
(b) $F\in\D_p(K_{5,\alpha})$ iff $r(F)<0$ and $\bar{F}(r(F)\exp(1/(.)))$ is regularly varying at $\infty$ with exponent $(-\alpha).$
\end{thm}
\begin{thm}\label{Ap_thm1}(Mohan and Ravi 1993)
A df $F\in\D_p(\Phi_1)$ if and only if $r(F) > 0,$ and
\[\lim_{t\uparrow r(F)}\frac{\bar{F}(t\exp(yu(t)))}{\bar{F}(t)}=\exp(-y)\]
for some positive valued function $u.$
\end{thm}
\begin{thm}\label{Ap_thm2}(Mohan and Ravi 1993)
A df $F\in\D_p(\Psi_1)$ if and only if $r(F) \leq 0,$ and
\[\lim_{t\uparrow r(F)}\frac{\bar{F}(t\exp(yu(t)))}{\bar{F}(t)}=\exp(y),\]
for some positive valued function $u.$
\end{thm}
\begin{thm}[Mohan and Subramanya (1998)] \label{thm_von}
Let df $F$ has pdf $f > 0$ on $(l(F), r(F)),$ and for some $\alpha>0,$
\begin{enumerate}
\item $\;F\in \mathcal D_p(K_{1,\alpha}),$ if $r(F)=\infty \quad \mbox{and} \quad \lim_{x\rightarrow \infty} \frac{x f(x)\log x}{1-F(x)}=\alpha.$
\item $\;F\in\mathcal D_p(K_{2,\alpha}),$ if
$ 0<r(F)<\infty \;\; \mbox{and} \;\; \lim_{x\rightarrow r(F)} \frac{xf(x)\log \left(\frac{r(F)}{x}\right) }{1-F(x)}=\alpha.$
\item $\;F\in\mathcal D_p(K_3),$ if
\begin{equation}\label{pvon3}
r(F)>0, \;\;  \int_{x}^{r(F)}\frac{(1-F(t))}{t}dt < \infty \; \mbox{and} \;
\lim_{x\rightarrow r(F)} \frac{x f(x)}{(1-F(x))^2}\int_{x}^{r(F)}\frac{(1-F(t))}{t}dt=1.\nonumber
\end{equation}

\item $\;F\in\mathcal D_p(K_{4,\alpha})$ if
$r(F)=0 \;\; \mbox{and} \;\; \lim_{x\uparrow 0} \frac{x f(x)\log (-x)}{1-F(x)}=\alpha.$

\item $\;F\in\mathcal D_p(K_{5,\alpha})$ if
$r(F)<0 \;\; \mbox{and} \;\; \lim_{x\uparrow r(F)} \frac{x f(x)\log \left(\frac{r(F)}{x}\right)}{1-F(x)}=\alpha.$

\item $\;F\in\mathcal D_p(K_6)$ if
\begin{equation}\label{pvon6}
r(F)\leq 0 ,\;\;   -\int_{-\infty }^{r(F)}\frac{(1-F(t))}{t}dt < \infty \; \mbox{and} \;\lim_{x\uparrow r(F)} \dfrac{x f(x)}{(1-F(x))^2}\int_{x}^{r(F)}\frac{1-F(t)}{t}dt=1.\nonumber
\end{equation}
\end{enumerate}
\end{thm}
\section{}\label{tab1}
\begin{table}[!ht]
\caption{Variance of .}
\begin{tabular}{lll}
\hline
Family& &Variance\\
\hline
$X^+$ &$\xi< 0$ &$e^{2\left(\mu-\sigma/\xi\right)}\sum_{j=0}^{2}{{2}\choose{j}}(-1)^j\left(M_Y\left(\frac{\sigma}{\abs{\xi}},\frac{1}{\abs{\xi}}\right)\right)^jM_Y\left(\frac{(2-j)\sigma}{\abs{\xi}},\frac{1}{\abs{\xi}}\right)$\\
$X^-$ & $\xi>0$ &$e^{2\left(\mu+\sigma/\xi\right)}\sum_{j=0}^{2}{{2}\choose{j}}(-1)^j\left(M_{Y^{-1}}\left(\frac{\sigma}{\xi},\frac{1}{\abs{\xi}}\right)\right)^jM_{Y^{-1}}\left(\frac{(2-j)\sigma}{\xi},\frac{1}{\abs{\xi}}\right)$\\ 
\hline
\end{tabular}
\end{table}
\section{}
\subsection{MCMC algorithm.}\label{Ape1}
\begin{enumerate}
\item[1.] Initialize the values at $\bm{\theta}^{(0)}=(\mu^{(0)},\sigma^{(0)},\xi^{(0)})$ and the counter at $j = 1.$ 
\item[2.] Simulate $\bm{\theta}^*\sim N(\bm{\theta}^{(j-1)},\omega_{\bm{\theta}}),$
where, $\omega_{\bm{\theta}}$ are chosen small enough.
\item[3.] Accept $\mu^{(j)}=\mu^*$ with probability $\Omega(\mu^*,\mu^{(j-1)})$ where,
\begin{eqnarray}
\Omega(\mu^*,\mu^{(j-1)})=\min\Big\{1,\frac{\pi(\mu^*|\eta^{(j-1)},\xi^{j-1)}) }{\pi(\mu^{(j-1)}|\eta^{(j-1)},\xi^{j-1)})}\Big\};\nonumber
\end{eqnarray}
And otherwise, $\mu^{(j)}=\mu^{(j-1)}.$
\item[4.] Accept $\eta^{(j)}=\eta^*$ with probability $\Omega(\eta^*,\eta^{(j-1)})$ where,
\begin{eqnarray}
\Omega(\eta^*,\eta^{(j-1)})=\min\Big\{1,\frac{\pi(\eta^*|\mu^{(j)},\xi^{(j-1)}) }{\pi(\eta^{(j-1)}|\mu^{(j)},\xi^{j-1)})}\Big\};\nonumber
\end{eqnarray}
And $\eta^{(j)}=\eta^{(j-1)}$ otherwise.
\item[5.] Accept $\xi^{(j)}=\xi^*$ with probability $\Omega(\xi^*,\xi^{(j-1)})$ where,
\begin{eqnarray}
\Omega(\xi^*,\xi^{(j-1)})=\min\Big\{1,\frac{\pi(\xi^*|\mu^{(j)},\eta^{(j)}) }{\pi(\xi^{(j-1)}|\mu^{(j)},\eta^{j)})}\Big\};\nonumber
\end{eqnarray}
And $\eta^{(j)}=\eta^{(j-1)}$ otherwise.
\item[6.] Increasing $j$ and return to step 2.
\end{enumerate}
\vspace{0.5in}
\subsection{Goodness of fit algorithm.}\label{Ape}
To test $H_0:X_1,\ldots,X_n\sim F(x;\theta),$ we can proceed as follows.
\begin{enumerate}
\item[1.] Compute $v_i=F(x_i;\hat{\theta}),$ where the $x_i$'s are in ascending order.
\item[2.] Compute $y_i=\eta^{\leftarrow}(v_i),$ where $\eta(\cdot)$ is the standard normal df and $\eta^{\leftarrow}(\cdot)$ its inverse;
\item[3.] Compute $u_i=\eta((y_i-\bar{y})/s_y),$ where $\bar{y}=\sum_{i=1}^{n}y_i/n$ and $s_y^2=\sum_{i=1}^{n}(y_i-\bar{y})^2/(n-1);$
\item[4.] Calculate
\[W^2=\sum_{i=1}^{n}\left(u_i-\frac{(2i-1)}{2n}\right)^2+\frac{1}{12n},\]
and
\[A^2=-n-\frac{1}{n}\sum_{i=1}^{n}((2i-1)\log(u_i)+(2n+1-2i)\log(1-u_i));\]
\item[5.] Modify $W^2$ into $C=W^2(1+0.5/n)$ and $A=A^2(1+0.75/n+2.25/n^2).$ Reject $H_0$ at the significance level $\alpha$ if the modified statistics exceed the upper tail significance points given in Table 1 of Chen and Balakrishnan (1995).
\end{enumerate}
\vspace{0.5in}
\begin{table}[!ht]
	\centering
	\caption{Goodness of fit tests}\label{Table1}
	\begin{tabular}{p{2cm} p{2cm} p{3cm} p{2cm} p{3cm}}
		\hline
		Laws& $C$ & &$A$ & \\
		\hline
		GEV &$0.0163
		$& $(\text{p-value}>0.5)$ & $0.4576$ & $(\text{p-value}>0.25)$\\
		Gumbel  &$0.0286
		$& $(\text{p-value}>0.5)$& $0.7544$& $(\text{p-value}>0.01)$\\
		PGEV  &$0.0327
		$ & $(\text{p-value}>0.5)$ &$0.9174$ & $(\text{p-value}>0.01)$\\
		\hline
	\end{tabular}
\end{table}

\newpage
\section{}
\begin{figure}[!ht]
\centering
\includegraphics[width=1\textwidth]{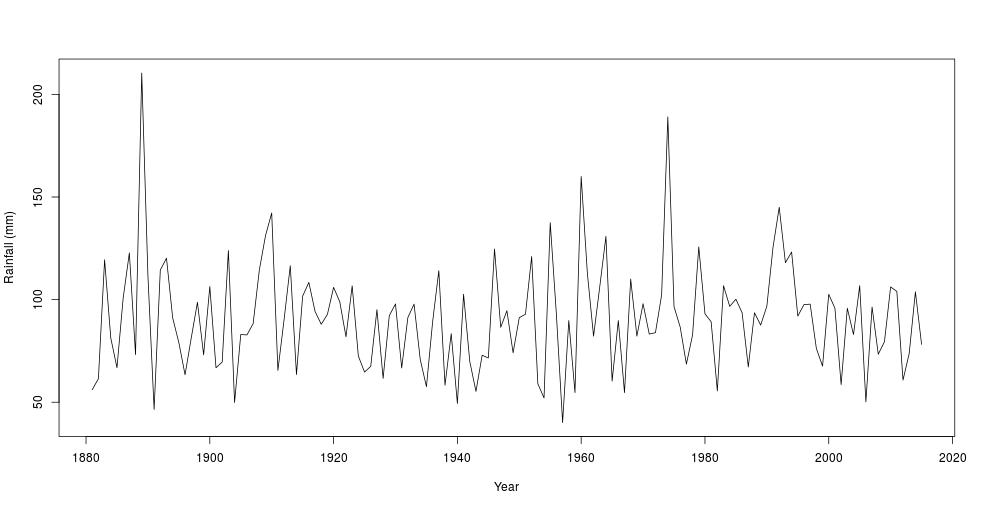}
\caption{Annual maximum rainfall recorded at Eudunda, Australia since
1881}\label{Data}
\end{figure}

\begin{figure}[!ht]
\centering
\includegraphics[width=1\textwidth]{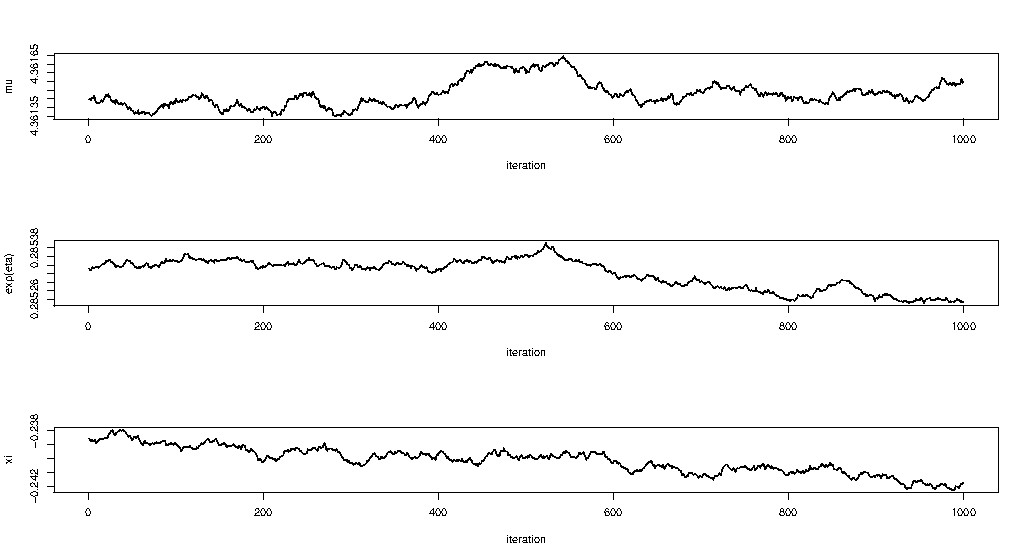}
\caption{MCMC calculations of  parameters in a Bayesian analysis of	the Eudunda, Australia annual maximum rainfall.}\label{graph_bayes}
\end{figure}

\begin{table}[!ht]
	\caption{The values parameters of EV families estimated}\label{Table2}
	\begin{tabular}{p{2cm} p{2cm} p{2cm}p{2cm}p{2cm}}
\hline
& $\mu$ & $\sigma$&  $\xi$&MLE \\ \hline
GEV  &$79.2669$ & $21.9150$ & $-0.0468$
&$-625.1091$\\
\vspace{0.1in}
&$(2.0814)$& $(1.4697)$ & $( 0.0503)$& \\
Gumbel & $78.7020$ & $21.6541$& $-$& $-625.4845$ \\
\vspace{0.1in}
&$(1.9672)$ &$(1.4248)$& &\\

PGEV &$4.3614$&$0.2853$& $-0.2386$&$-626.3673$\\
\vspace{0.1in}
&$(0.0265)$&$(0.0179)$&$( 0.0372)$&\\
\hline
\end{tabular}
\end{table}

\vspace{0.5in}

\begin{table}[!ht]
		\centering
\caption{Return level values for  }\label{Table3}
\begin{tabular}{p{3cm} ccccccc}
\hline
Return Period&$4$&$10$&$15$&$20$&$30$&$35$&$50$\\
Return level(mm)&$107$&$129$&$138$&$144$&$152$&$156$&$162$\\
			\hline
	\end{tabular}
\end{table}

\begin{figure}[!ht]
\centering
\includegraphics[width=1\textwidth]{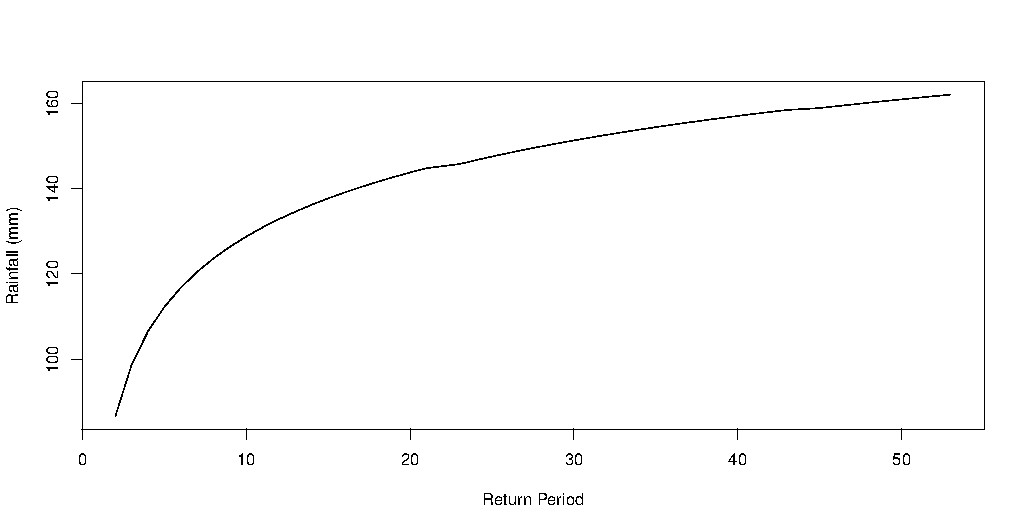}
\caption{Predictive return levels $x_p$ against $p=1/m.$}\label{return_graph}
\end{figure}


\end{document}